\documentclass[10pt,conference]{IEEEtran}
\usepackage{fixltx2e}
\usepackage{cite}
\usepackage{url}
\usepackage{mathrsfs}
\usepackage{color}
\usepackage{float}
\usepackage[caption=false]{subfig}

\ifCLASSINFOpdf
   \usepackage[pdftex]{graphicx}
   \graphicspath{{Figs/}}
   \DeclareGraphicsExtensions{.pdf,.jpeg,.png}
\else
\fi

\usepackage[cmex10]{amsmath}
\usepackage{amsmath}
\usepackage{amssymb}
\usepackage{amsthm}
\usepackage{amsfonts}
\usepackage{bm}
\usepackage{xfrac}
\usepackage{empheq}
\usepackage[normalem]{ulem} 
\usepackage{soul} 
\usepackage{mathtools}
\DeclarePairedDelimiter\ceil{\lceil}{\rceil}

\newtheorem{theorem}{Theorem}

\newtheorem{example}{Example}
\newtheorem{proposition}{Proposition}
\newtheorem{lemma}{Lemma}

\newtheorem{corollary}{Corollary}
\newtheorem{remark}{Remark}
\theoremstyle{definition}

\usepackage[a4paper,bindingoffset=0.2in,%
left=1in,right=1in,top=1in,bottom=1in,%
footskip=.25in]{geometry}
\graphicspath{{figs/}}

\begin{document}
	\newgeometry{left=0.7in,right=0.7in,top=.5in,bottom=1in}
	\title{Private Variable-Length Coding with Zero Leakage}
\vspace{-5mm}
\author{
	\IEEEauthorblockN{Amirreza Zamani$^\dagger$, Tobias J. Oechtering$^\dagger$, Deniz G\"{u}nd\"{u}z$^\ddagger$, Mikael Skoglund$^\dagger$ \vspace*{0.5em}
		\IEEEauthorblockA{\\
			$^\dagger$Division of Information Science and Engineering, KTH Royal Institute of Technology \\
			$^\ddagger$Dept. of Electrical and Electronic Engineering, Imperial College London\\
			Email: \protect amizam@kth.se, oech@kth.se, d.gunduz@imperial.ac.uk, skoglund@kth.se }}
}
\maketitle
%
\begin{abstract}
	 A private compression design problem is studied, where an encoder observes useful data $Y$, wishes to compress it using variable length code and communicates it through an unsecured channel. Since $Y$ is correlated with private attribute $X$, the encoder uses a private compression mechanism to design encoded message $\cal C$ and sends it over the channel. An adversary is assumed to have access to the output of the encoder, i.e., $\cal C$, and tries to estimate $X$.
	Furthermore, it is assumed that both encoder and decoder have access to a shared secret key $W$. 
	The design goal is to encode message $\cal C$ with minimum possible average length that satisfies a perfect privacy constraint. 
	To do so we first consider two different privacy mechanism design problems and find upper bounds on the entropy of the optimizers by solving a linear program. We use the obtained optimizers to design $\cal C$. 
	In two cases we strengthen the existing bounds: 1. $|\mathcal{X}|\geq |\mathcal{Y}|$; 2. The realization of $(X,Y)$ follows a specific joint distribution. In particular, considering the second case we use two-part construction coding to achieve the upper bounds. Furthermore, in a numerical example we study the obtained bounds and show that they can improve the existing results. 
\end{abstract}
\section{Introduction}
In this paper, random variable (RV) $Y$ denotes the useful data and is correlated with the private attribute denoted by RV $X$. An encoder wishes to compress $Y$ and communicate it with a user over an unsecured channel. The encoded message is described by RV $\cal{C}$.  As shown in Fig.~\ref{ITWsys}, it is assumed that an adversary has access to the encoded message $\mathcal{C}$, and wants to extract information about $X$. Moreover, it is assumed that the encoder and decoder have access to a shared secret key denoted by RV $W$ with size $M$. The goal is to design encoded message $\cal C$, which compresses $Y$, using a variable length code with minimum possible average length that satisfies a perfect privacy constraint. We utilize techniques used in privacy mechanism and compression design problems and combine them to build such $\cal C$.
In this work, considering two different cases we extend previous existing results \cite{kostala,kostala2}. 

Recently, the privacy mechanism and compression design problems are receiving increased attention
\cite{shannon, dworkal, dwork1, gunduz2010source, schaefer, sankar, yamamoto1988rate, Calmon2, yamamoto, issa, makhdoumi,borz,khodam,Khodam22,kostala, kostala2, calmon4, issajoon, asoo, Total, issa2, king1, king2, kosenaz, 9457633,asoodeh1,bassi,king3,deniz3}. 
Specifically, in \cite{shannon}, a notion of perfect secrecy is introduced by Shannon where the public data and private data are statistically independent. Furthermore, Shannon cipher system is analyzed where one of $M$ messages is sent over a channel wiretapped by an eavesdropper who tries to guess the message. As shown in \cite{shannon}, perfect secrecy is achievable if and only if the shared secret key length is at least $M$. Perfect secrecy, where public data and private data are independent, is also used in \cite{dworkal,dwork1}, where sanitized version of a data base is disclosed for a public use. Equivocation as a measure of information leakage for information theoretic security has been used in \cite{gunduz2010source, schaefer, sankar}. A rate-distortion approach to information theoretic secrecy is studied in \cite{yamamoto1988rate}.   
Fundamental limits of the privacy utility trade-off measuring the leakage using estimation-theoretic guarantees are studied in \cite{Calmon2}. A related source coding problem with secrecy is studied in \cite{yamamoto}.
The concept of maximal leakage has been introduced in \cite{issa} and some bounds on the privacy utility trade-off have been derived. 
The concept of privacy funnel is introduced in \cite{makhdoumi}, where the privacy utility trade-off has been studied considering the log-loss as privacy measure and a distortion measure for utility. 
The privacy-utility trade-offs considering equivocation and expected distortion as measures of privacy and utility are studied in both \cite{sankar} and \cite{yamamoto}.

In \cite{borz}, the problem of privacy-utility trade-off considering mutual information both as measures of utility and privacy is studied. It is shown that under the perfect privacy assumption, the privacy mechanism design problem can be obtained by a linear program. 
\begin{figure}[]
	\centering
	\includegraphics[scale = .1]{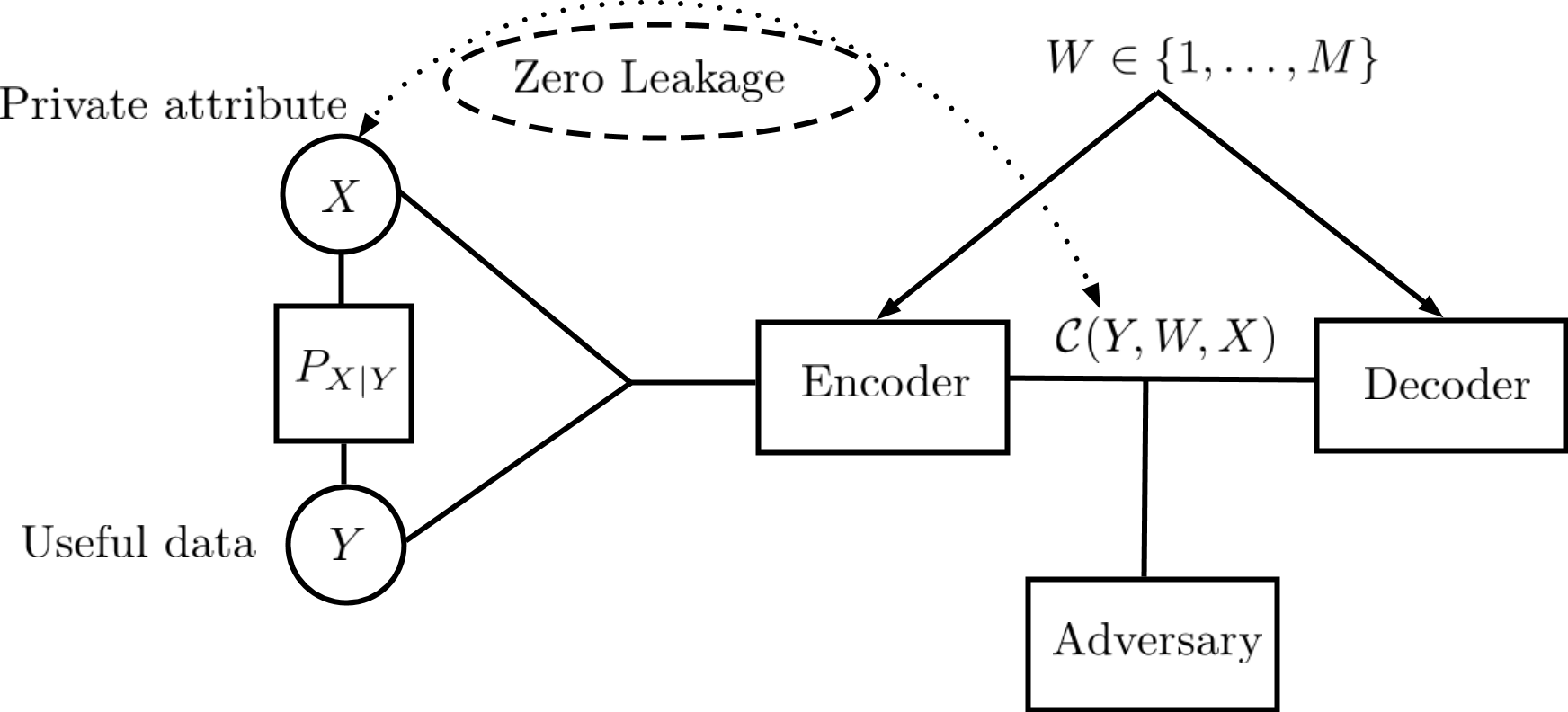}
	\caption{In this work an encoder wants to compress $Y$ which is correlated with $X$ under a zero privacy leakage constraint and send it over a channel where an eavesdropper has access to the output of the encoder. The encoder and decoder have an advantage of using shared secret key.}
	\label{ITWsys}
\end{figure}
Moreover, in \cite{borz}, it has been shown that information can be only revealed if the kernel (leakage matrix) between useful data and private data is not invertible. In \cite{khodam}, the work \cite{borz} is generalized by relaxing the perfect privacy assumption allowing some small bounded leakage. More specifically, the privacy mechanisms with a per-letter privacy criterion considering an invertible kernel are designed allowing a small leakage. This result is generalized to a non-invertible leakage matrix in \cite{Khodam22}.\\
In \cite{kostala}, an approach to partial secrecy that is called \emph{secrecy by design} has been introduced and is applied to two information processing problems: privacy mechanism design and lossless compression. For the privacy design problem, bounds on privacy-utility trade-off are derived by using the Functional Representation Lemma. These results are derived under the perfect secrecy assumption.
In \cite{king1}, the privacy problems considered in \cite{kostala} are generalized by relaxing the perfect secrecy constraint and allowing some leakages. 
In \cite{king2}, the privacy-utility trade-off with two different per-letter privacy constraints is studied. 
Moreover, in \cite{kostala}, the problems of fixed length and variable length compression have been studied and upper and lower bounds on the average length of encoded message have been derived. These results are derived under the assumption that the private data is independent of the encoded message. A similar approach has been used in \cite{kostala2}, where in a lossless compression problem the relations between secrecy, shared key and compression considering perfect secrecy, secrecy by design, maximal leakage, mutual information leakage and local differential privacy have been studied.

In this work we consider the problem in \cite{kostala}, where for the problem of lossless data compression, strong information theoretic guarantees are provided and fundamental limits are characterized when the private data is a deterministic function of the useful data. In this paper, we improve the bounds obtained in \cite{kostala}. 

To this end we combine the privacy design techniques used in \cite{king3}, which are based on extended versions of Functional Representation Lemma (FRL) and Strong Functional Representation Lemma (SFRL), as well as the lossless data compression design in \cite{kostala}. We find lower and upper bounds on the average length of the encoded message $\cal C$ and study them in different scenarios.
Considering a specific set of joint distributions for $(X,Y)$, we propose an algorithm with low complexity to find bounds on the optimizer of privacy problems in \cite{king3} with zero leakage. We use the obtained results to design $\cal C$ which achieves tighter bounds compared to \cite{kostala}. Furthermore, when $|\mathcal{X}|\geq |\mathcal{Y}|$, we propose a new code that improves the bounds in \cite{kostala}. Finally, in a numerical example we study the bounds and compare them with \cite{kostala}.  

\section{system model and Problem Formulation} \label{sec:system}
Let $P_{XY}$ denote the joint distribution of discrete random variables $X$ and $Y$ defined on alphabets $\cal{X}$ and $\cal{Y}$. We assume that cardinality $|\mathcal{X}|$ is finite and $|\mathcal{Y}|$ is finite or countably infinite.
We represent $P_{XY}$ by a matrix defined on $\mathbb{R}^{|\mathcal{X}|\times|\mathcal{Y}|}$ and
marginal distributions of $X$ and $Y$ by vectors $P_X$ and $P_Y$ defined on $\mathbb{R}^{|\mathcal{X}|}$ and $\mathbb{R}^{|\mathcal{Y}|}$ given by the row and column sums of $P_{XY}$.  
The relation between $X$ and $Y$ is given by the leakage matrix $P_{X|Y}$ defined on $\mathbb{R}^{|\mathcal{X}|\times|\cal{Y}|}$.
The shared secret key is denoted by discrete RV $W$ defined on $\{1,..,M\}$ and is assumed to be accessible by both encoder and decoder. Furthermore, we assume that $W$ is uniformly distributed and is independent of $X$ and $Y$. 
A prefix-free code with variable length and shared secret key of size $M$ is a pair of mappings:
\begin{align*}
(\text{encoder}) \ \mathcal{C}: \ \mathcal{Y}\times \{1,..,M\}\rightarrow \{0,1\}^*\\
(\text{decoder}) \ \mathcal{D}: \ \{0,1\}^*\times \{1,..,M\}\rightarrow \mathcal{Y}.
\end{align*}
The output of the encoder $\mathcal{C}(Y,W)$ describes the encoded message. 
Since the code is prefix free, any codeword in the image of $\cal C$ is not a prefix of other codewords. The variable length code $(\mathcal{C},\mathcal{D})$ is lossless if 
\begin{align}
\mathbb{P}(\mathcal{D}(\mathcal{C}(Y,W),W)=Y)=1.
\end{align}  
We define perfectly-private code.
The code $(\mathcal{C},\mathcal{D})$ is \textit{perfectly-private} if
\begin{align}
I(\mathcal{C}(Y,W);X)=0.\label{lash1}
\end{align}
Let $\xi$ be the support of $\mathcal{C}(W,Y)$ and 
for any $c\in\xi$ let $\mathbb{L}(c)$ be the length of the codeword. The code $(\mathcal{C},\mathcal{D})$ is \textit{$(\alpha,M)$-variable-length} if 
\begin{align}\label{jojo}
\mathbb{E}(\mathbb{L}(\mathcal{C}(Y,w)))\leq \alpha,\ \forall w\in\{1,..,M\}.
\end{align} 
Finally, let us define the sets $\mathcal{H}(\alpha,M)$ as follows: 
$\mathcal{H}(\alpha,M)\triangleq\{(\mathcal{C},\mathcal{D}): (\mathcal{C},\mathcal{D})\ \text{is}\ \text{perfectly-private and}\ (\alpha,M)\text{-variable-length}  \}$, 
The private compression design problems can be stated as follows
\begin{align}
\mathbb{L}(P_{XY},M)&=\inf_{\begin{array}{c} 
	\substack{(\mathcal{C},\mathcal{D}):(\mathcal{C},\mathcal{D})\in\mathcal{H}(\alpha,M)}
	\end{array}}\alpha,\label{main1}
\end{align}
Furthermore, let us recall the privacy mechanism design problems considered in \cite{king1} with zero leakage as follows
\begin{align}
g_{0}(P_{XY})&=\max_{\begin{array}{c} 
	\substack{P_{U|Y}:X-Y-U\\ \ I(U;X)=0,}
	\end{array}}I(Y;U),\label{maing}\\
h_{0}(P_{XY})&=\max_{\begin{array}{c} 
	\substack{P_{U|Y,X}: I(U;X)=0,}
	\end{array}}I(Y;U).\label{mainh}
\end{align} 
Finally, we define a set of joint distributions $\hat{\mathcal{P}}_{XY}$ as follows
\begin{align}
\hat{\mathcal{P}}_{XY}\triangleq \{P_{XY}:g_{0}(P_{XY})=h_{0}(P_{XY})\}.
\end{align}
 \section{Main Results}\label{sec:resul}
 In this section, we first study the properties of $\hat{\mathcal{P}}_{XY}$ and provide a sufficient condition on $P_{XY}$ so that it belongs to the corresponding set. We then find upper bounds on the entropy of the optimizers of $g_{0}(P_{XY})$ and $h_{0}(P_{XY})$. The obtained bounds help us to find upper bounds on $\mathbb{L}(P_{XY},M)$. In other words, for $P_{XY}\in\hat{\mathcal{P}}_{XY}$ we use the solutions of $g_{0}(P_{XY})$ and $h_{0}(P_{XY})$ to design $\mathcal{C}$, however in \cite{kostala}, the design of the code is based on Functional Representation Lemma. Finally, we provide lower bounds on $\mathbb{L}(P_{XY},M)$, study the bounds in a numerical example and compare them with \cite{kostala}.
 
 In the next lemma, let $C(X,Y)$ denote the common information between $X$ and $Y$, where common information corresponds to the Wyner \cite{wyner} or G{\'a}cs-K{\"o}rner \cite{gacs1973common} notions of common information.
 \begin{lemma}\label{set}
 	If $C(X,Y)=I(X;Y)$, then $P_{XY}\in\hat{\mathcal{P}}_{XY}$.
 \end{lemma}
 \begin{proof}
 	The proof is provided in \cite[Proposition 6]{king3} with $\epsilon=0$ and follows similar arguments as the proof of \cite[Theorem 2]{7888175}.
 \end{proof}
 \begin{corollary}
 	If $X$ is a deterministic function of $Y$, then $P_{XY}\in\hat{\mathcal{P}}_{XY}$, since in this case we have $C(X,Y)=I(X;Y)$. Furthermore, this result has also been shown in \cite[Proposition 6]{king3}.  
 \end{corollary}
In the next result we provide properties of the optimizers for $g_{0}(P_{XY})$ and $h_{0}(P_{XY})$.
 \begin{lemma}\label{mohem}
 	If $g_{0}(P_{XY})=h_{0}(P_{XY})$ then
 	\begin{align}\label{1}
 	g_{0}(P_{XY})=h_{0}(P_{XY})=H(Y|X),
 	\end{align}
 	and both $g_{0}(P_{XY})$ and $h_{0}(P_{XY})$ have the same set of optimizers. Furthermore, for any optimizer $U^*$ we have 
 	\begin{align}
 	H(Y|U^*,X)=0,\label{2}\\
 	I(X;U^*|Y)=0,\label{3}\\
 	I(X;U^*)=0.\label{4}
 	\end{align}  
 	Finally, if $U$ satisfies \eqref{2}, \eqref{3}, and \eqref{4}, then it is an optimizer for both $g_{0}(P_{XY})$ and $h_{0}(P_{XY})$.
 \end{lemma}
 \begin{proof}
 	All the statements can be shown using \cite[Theorem 7]{kostala} and the key equation as follows
 	\begin{align}\label{key}
 	I(U;Y)\!=\!I(X;U)\!+\!H(Y|X)\!-\!I(X;U|Y)\!-\!H(Y|X,U).
 	\end{align}
 \end{proof}
\begin{remark}
	The optimization problems in \eqref{maing} and \eqref{mainh} do not have unique optimizers. For instance, let $X=f(Y)$, by using the constructions used in \cite{kostala} and \cite{kosnane} we can attain the optimum value.
\end{remark}
 Next, we show that if $P_{XY}\in\hat{\mathcal{P}}_{XY}$, without loss of optimality we can assume that $|\mathcal{U}|\leq \text{null}(P_{X|Y})+1$, where $\text{nul}(P_{X|Y})$ corresponds to the dimension of the null space of $P_{X|Y}$.  
 \begin{lemma}\label{choon}
 	For any $P_{XY}$ we have
 	\begin{align}\label{jensi}
 	g_{0}(P_{XY})= \max_{\begin{array}{c} 
 		\substack{P_{U|Y}:X-Y-U,\\ I(X;U)=0,\\ |\mathcal{U}|\leq\ \text{null}(P_{X|Y})+1}
 		\end{array}}I(Y;U).
 	\end{align}
 	Furthermore, let $P_{XY}\in\hat{\mathcal{P}}_{XY}$. Then
 	\begin{align}
 	h_{0}(P_{XY})&=g_{0}(P_{XY})\nonumber\\&= \max_{\begin{array}{c} 
 		\substack{P_{U|Y}:X-Y-U\\ \ I(X;U)=0,\\ |\mathcal{U}|\leq\ \text{null}(P_{X|Y})+1}
 		\end{array}}I(Y;U)\label{he}\\&= \max_{\begin{array}{c} 
 		\substack{P_{U|Y}: I(X;U)=0,\\ |\mathcal{U}|\leq\ \text{null}(P_{X|Y})+1}
 		\end{array}}I(Y;U)\label{hehe}
 	\end{align}
 \end{lemma}
 \begin{proof}
 	The proof of \eqref{jensi} is provided in \cite[Theorem 1]{borz}. It only remains to show the equality $\eqref{he}=\eqref{hehe}$. Let $U^*$ be an optimizer of \eqref{he}. Using \eqref{jensi}, it is also an optimizer of $g_0(P_{XY})$ and hence $h_0(P_{XY})$. Thus, $\eqref{he}=\eqref{hehe}$. In other words, for $P_{XY}\in\hat{\mathcal{P}}_{XY}$, we can assume $|\mathcal{U}|\leq \text{null}(P_{X|Y})+1$ in both problems $g_0(P_{XY})$ and $h_0(P_{XY})$.   
 \end{proof}
 Before stating the next result let us define a set $\mathcal{U}^1(P_{XY})$ and a function $\mathcal{K}(P_{XY})$ as follows
 \begin{align}
 \mathcal{U}^1(P_{XY})&\triangleq \{U: U\ \text{satisfies \eqref{2}, \eqref{3}, \eqref{4}}\}\\
 \mathcal{K}(P_{XY})&\triangleq \min_{U\in \mathcal{U}^1(P_{XY})} H(U).
 \end{align}
 Note that $ \mathcal{U}^1(P_{XY})$ can be empty for some joint distribution $P_{XY}$, however using Lemma \ref{mohem} when $P_{XY}\in\hat{\mathcal{P}}_{XY}$ it is the set which contains all the optimizers satisfying $g_0(P_{XY})=h_0(P_{XY})$.  
 Furthermore, the function $\mathcal{K}(P_{XY})$ finds the minimum entropy of all optimizers satisfying $g_0(P_{XY})=h_0(P_{XY})$.  
 \begin{lemma}\label{choon1}
 	Let $P_{XY}\in\hat{\mathcal{P}}_{XY}$. We have
 	\begin{align}
 	\mathcal{K}(P_{XY})\leq \log(\text{null}(P_{X|Y})+1).
 	\end{align}
 \end{lemma} 
 \begin{proof}
 	The proof follows from Lemma \ref{choon}. Let $U^*$ be any optimizer of $g_0(P_{XY})$ that satisfies $|\mathcal{U}^*|\leq \text{null}(P_{X|Y})+1$. We have $\mathcal{K}(P_{XY})\leq  H(U^*) \leq\log(\text{null}(P_{X|Y})+1)$.
 \end{proof}
 \begin{remark}
 	Lemma \ref{choon1} asserts that when $P_{XY}\in\hat{\mathcal{P}}_{XY}$, there exist a $U$ that satisfies \eqref{2}, \eqref{3}, \eqref{4}, and 
 	\begin{align}
 	H(U)\leq \log(\text{null}(P_{X|Y})+1).
 	\end{align} 
 \end{remark}
\begin{remark}
	The upper bound obtained in Lemma \ref{choon1} can be significantly smaller than the upper bound found in Lemma \cite[Lemma 2]{kostala}. This helps us to find codes that have less average length compared to \cite{kostala}.
\end{remark}
 Before stating the next result we define 
 \begin{align}
 A_{XY}&\triangleq \begin{bmatrix}
 &P_{y_1}-P_{y_1|x_1} &\ldots & P_{y_{|\mathcal{Y}|}}-P_{y_{q}|x_1}\\
 	&\cdot &\ldots &\cdot\\
 	&P_{y_1}-P_{y_1|x_{t}} &\ldots & P_{y_{q}}-P_{y_{q}|x_{t}}
 \end{bmatrix}\!\in\! \mathbb{R}^{t\times q},\\
 b_{XY}&\triangleq \begin{bmatrix}
 H(Y|x_1)-H(Y|X) \\
 \cdot \\
 H(Y|x_t)-H(Y|X)
 \end{bmatrix}\in\mathbb{R}^{t},\ \bm{a}\triangleq \begin{bmatrix}
 a_1 \\
 \cdot \\
 a_q
 \end{bmatrix}\in\mathbb{R}^{q}.
 \end{align}
 where $t=|\mathcal{X}|$ and $q=|\mathcal{Y}|$.
 \begin{theorem}\label{11}
 	Let $P_{XY}\in\hat{\mathcal{P}}_{XY}$. For any $U\in\mathcal{U}^1$ we have
 	\begin{align}
 	H(Y|X)+\!\!\!\!\!\min_{a_i:A_{XY}\bm{a}=b_{XY},\bm{a}\geq 0}\sum_{i=1}^{q} P_{y_i}a_i\leq \mathcal{K}(P_{XY})\leq \label{lower}\\
 	 H(U)\leq H(Y|X)+\!\!\!\!\!\max_{a_i:A_{XY}\bm{a}=b_{XY},\bm{a}\geq 0}\sum_{i=1}^{q} P_{y_i}a_i.\label{tala}
 	\end{align}
 	The upper bound can be strengthened and we have
 	\begin{align}\label{koontala}
 	\mathcal{K}(P_{XY}) &\leq H(Y|X)+\!\!\!\!\!\!\!\!\!\!\!\!\!\!\!\!\!\!\!\!\!\!\!\!\!\!\!\!\!\max_{\begin{array}{c}
 		\substack{a_i:A_{XY}\bm{a}=b_{XY},\bm{a}\geq 0,\\
 		\sum_{i=1}^{q} P_{y_i}a_i\leq \log(\text{null}(P_{X|Y})+1)-H(Y|X)} \end{array} }\!\!\!\!\sum_{i=1}^{q} P_{y_i}a_i\\&\leq \log(\text{null}(P_{X|Y})+1).\label{khari}
 	\end{align}
 	Moreover, when $\text{rank}(A_{XY})=|\mathcal{Y}|$ and $Y\neq f(X)$, we have $|\mathcal{U}^1|=1$ and for the unique optimizer we have 
 	\begin{align}
 	\mathcal{K}(P_{XY})=H(Y|X)+\sum_i P_{y_i}a_i
 	\end{align}
 	where $A_{XY}\bm{a}=b_{XY}, \bm{a}\geq 0$ and $\sum_{i=1}^{q} P_{y_i}a_i\leq \log(\text{null}(P_{X|Y})+1)-H(Y|X)$.
 \end{theorem}
\begin{proof}
	Let $U \in \mathcal{U}^1$. In the following we expand $H(Y,U|X=x)$ using chain rule:
	\begin{align}
	H(Y,U|X=x)&=H(U|X=x)+H(Y|U,X=x)\nonumber\\&\stackrel{(a)}{=}H(U),\label{halal}
	\end{align}   
	where in (a) we used 
	\begin{align*}
	I(X;U)=0\rightarrow H(U)=H(U|X)=H(U|X=x),
	\end{align*}
	and 
	\begin{align*}
	H(Y|X,U)=0 \rightarrow H(Y|U,X=x)=0.
	\end{align*}
	Furthermore,
	\begin{align}
	H(Y,U|X=x)&=H(Y|X=x)+H(U|Y,X=x)\nonumber\\&=\!\! H(Y|X\!=\!x)\!+\!\!\!\sum_y\!\! P_{y|x}H(U|Y\!=\!y,\!X\!=\!x)\nonumber\\&\stackrel{(b)}{=}\!\! H(Y|X\!=\!x)\!+\!\sum_y\!\! P_{y|x}H(U|Y\!=\!y),\label{haroom}
	\end{align}
	where (b) follows by the Markov chain $X-Y-U$. Combining \eqref{halal} and \eqref{haroom}, for any $x\in\mathcal{X}$ we have
	\begin{align}\label{really}
	H(U)= H(Y|X\!=\!x)\!+\!\sum_y\! P_{y|x}H(U|Y\!=\!y).
	\end{align}
	Furthermore, by using Lemma \ref{mohem}, $U$ is an optimizer of both $g_{0}(P_{XY})$ and $h_{0}(P_{XY})$. Thus, 
	\begin{align}\label{shobhe}
	H(U)\!=\!H(Y|X)\!+\!U(U|Y)\!=\!H(Y|X)\!+\!\!\!\sum_y\! P_yH(U|Y=y)
	\end{align} 
	Combining \eqref{shobhe} and \eqref{really} we have
	\begin{align}
	&H(Y|X)\!+\!\!\!\sum_y\! P_yH(U|Y=y)=\\&H(Y|X\!=\!x)\!+\!\sum_y\! P_{y|x}H(U|Y\!=\!y)
	\end{align}
	Thus, for any $x\in \mathcal{X}$
	\begin{align}\label{jaleb}
	\sum_y (P_{y|x}-P_y)H(U|Y\!=\!y)=H(Y|X)-H(Y|X=x).
	\end{align}
	For any optimizer $U \in \mathcal{U}^1$, \eqref{jaleb} holds, hence, if we take maximum over $H(U|Y=y)$ and by letting $a_i=H(U|Y=y_i)$, we obtain
	\begin{align}
	H(U)\leq \max H(U)=H(Y|X)+\max \sum_i P_{y_i}a_i
	\end{align}
	subject to
	\begin{align}
	A_{XY}\bm{a}=b_{XY},\ \bm{a}\geq 0.
	\end{align}
	To prove \eqref{koontala}, note that since we consider all optimizers of $g_0(P_{XY})$, there exist at least one $U^*$ that satisfies $|\mathcal{U}^*|\leq \text{null}(P_{X|Y})+1$ which results $H(U^*)\leq \log(\text{null}(P_{X|Y})+1)$. Since \eqref{jaleb} must hold for any optimizer, $U^*$ satisfies the inequality in \eqref{koontala}. Moreover, \eqref{khari} follows by Lemma \ref{choon1} and the constraint we added to have $H(U)\leq \log(\text{null}(P_{X|Y})+1)$. To prove the last statement note that when $Y=f(X)$, using the key equation in \eqref{key}, the upper bound on $h_0(P_{XY})$ equals zero, which results in $h_0(P_{XY})=g_0(P_{XY})=0$. Furthermore, as shown in \cite[Theorem 5]{kostala}, when $Y\neq f(X)$ we have $h_0(P_{XY})>0$, thus, the system of linear equations in \eqref{jaleb} has at least one solution. However, due to $\text{rank}(A_{XY})=|\mathcal{Y}|$ it has a unique solution, since number of independent equations is at least number of variables. Finally, to prove the lower bound in \eqref{lower}, note that for any $U$ that satisfies \eqref{2}, \eqref{3}, and \eqref{4}, \eqref{jaleb} holds. By taking minimum over $H(U)=H(Y|X)+\sum_i P_{y_i}H(U|y_i)$ subject to \eqref{jaleb}. Moreover, the inequality holds for $U^*$ that achieves the minimum entropy.  
\end{proof}
\begin{remark}\label{ant}
	As shown in \cite{borz}, $g_0(P_{XY})$ can be obtained by solving a linear program in which the size of the matrix in the system of linear equations is at most $|\mathcal{Y}|\times\binom{|\mathcal{Y}|}{\text{rank}(P_{X|Y})}$ with at most $\binom{|\mathcal{Y}|}{\text{rank}(P_{X|Y})}$ variables. By solving the linear program as proposed in \cite{borz} we can find the exact value of $\mathcal{K}(P_{XY})$ and the joint distribution $P_{U|YX}$ that achieves it. The complexity of the linear program in \cite{borz} can grow faster than exponential functions with respect to $|\mathcal{Y}|$, however the complexity of our proposed method grows linearly with $|\mathcal{Y}|$. Thus, our proposed upper bound has less complexity compared to the solution in \cite{borz}. Furthermore, in a special case where $\text{rank}(A_{XY})=|\mathcal{Y}|$ we can find the exact value of $\mathcal{K}(P_{XY})$ using our method. One necessary condition for $\text{rank}(A_{XY})=|\mathcal{Y}|$ is to have $|\mathcal{X}|\geq |\mathcal{Y}|+1$, since the summation of rows in $A_{XY}$ equals to zero.   
\end{remark}
\begin{remark}
	The optimizer of $g_0(P_{XY})$ does not help us to build $\cal C$, since the constraint $H(Y|U,X)=0$ does not hold in general, however, the optimizer of $h_0(P_{XY})$ satisfies it. Hence, we consider the cases where the equality $g_0(P_{XY})=h_0(P_{XY})$ holds. In this case, $H(Y|U,X)=0$ and we have tighter bounds on $H(U)$ compared to \cite{kostala}.  
\end{remark}
\begin{remark}
	The upper bound in \eqref{tala} holds for any $U\in\mathcal{U}^1$, however, the upper bound in \eqref{tala} asserts that there exists $U\in\mathcal{U}^1$ such that the bound holds. The lower bound in \eqref{lower} can be used to find lower bound on $\mathbb{L}(P_{XY},M)$. Similar to Remark \ref{ant}, the linear program obtaining the lower bound in \eqref{lower} has less complexity than \cite{borz}. 
\end{remark}
Next, we obtain lower and upper bounds on $\mathbb{L}(P_{XY},M)$.
\begin{theorem}\label{loo}
	For the pair of RVs $(X,Y)$ let $P_{XY}\in\hat{\mathcal{P}}_{XY}$ and the shared secret key size be $|\mathcal{X}|$, i.e., $M=|\mathcal{X}|$. Furthermore, let $q=|\mathcal{Y}|$ and $\beta=\log(\text{null}(P_{X|Y})+1)$. Then, we have 
	\begin{align}
	&\mathbb{L}(P_{XY},|\mathcal{X}|)
		\leq \mathcal{K}(P_{XY})+1+\ceil{\log(|\mathcal{X}|)} \label{log}\\
	&\leq H(Y|X)\!\!+\!\!\!\!\!\!\!\!\!\!\!\!\!\!\!\!\!\!\!\!\max_{\begin{array}{c}
		\substack{a_i:A_{XY}\bm{a}=b_{XY},\bm{a}\geq 0,\\
			\sum_{i=1}^{q}\! P_{y_i}a_i\leq \beta-H(Y|X)} \end{array} }\!\!\sum_{i=1}^{q} \!\!P_{y_i}a_i\!+\!1\!+\!\ceil{\log(|\mathcal{X}|)}\label{ass}\\&\leq \beta+1+\!\ceil{\log(|\mathcal{X}|)}\label{mass}
	\end{align}
	For any $P_{XY}$ we have
	\begin{align}
	&\mathbb{L}(P_{XY},|\mathcal{X}|)\leq \min_{\begin{array}{c} 
		\substack{P_{U|Y,X}: I(U;X)=0,\\ H(Y|X,U)=0}
		\end{array}}H(U)+1+\ceil{\log(|\mathcal{X}|)}\label{kos2}\\
	&\leq 1+\min(\sum_x H(Y|X=x),\ceil{\log(|\mathcal{X}|(|\mathcal{Y}|\!-\!1)\!+\!1)\!-\!1})\nonumber\\&+\ceil{\log(|\mathcal{X}|)}\label{kos1},
	\end{align}
	and if $X=f(Y)$,
	\begin{align}\label{govad}
	\mathbb{L}(P_{XY},|\mathcal{X}|)\leq \ceil{\log(|\mathcal{Y}|\!-\!|\mathcal{X}|\!+\!1)\!}\!+\!\ceil{\log(|\mathcal{X}|)}.
	\end{align} 
	Finally, for any $P_{XY}$ with $|\mathcal{Y}|\leq |\mathcal{X}|$ we have
	\begin{align}
	\mathbb{L}(P_{XY},|\mathcal{Y}|)\leq \ceil{\log{|\mathcal{Y}|}}.\label{kos3} 
	\end{align}
\end{theorem}
\begin{figure}[]
	\centering
	\includegraphics[scale = .12]{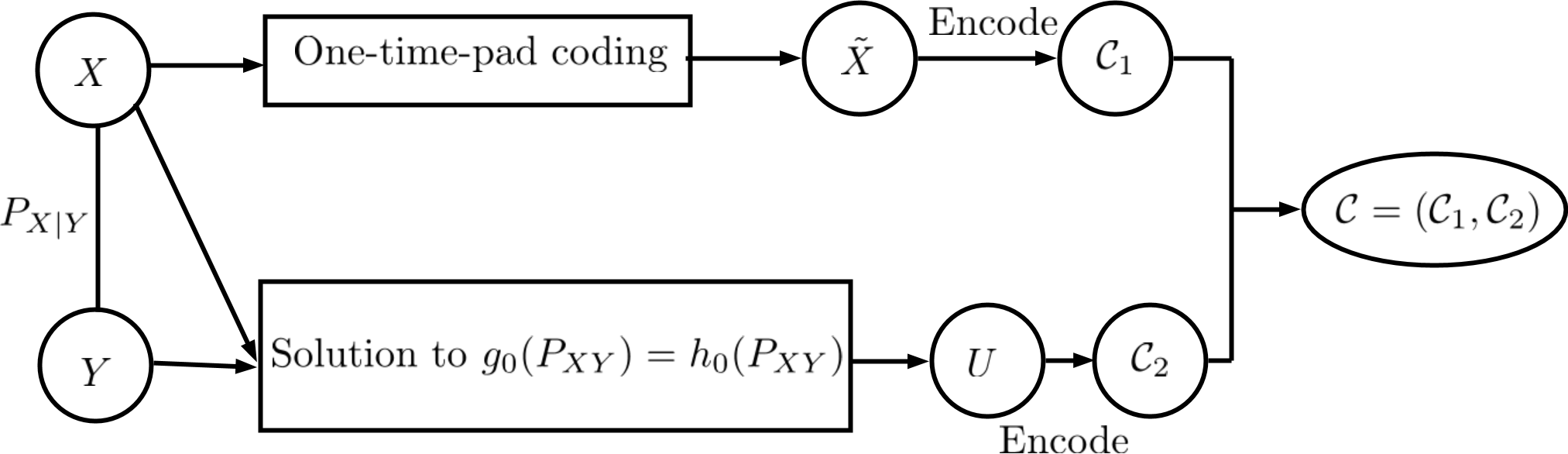}
	\caption{In this work we use two-part construction coding strategy to send codewords over the channels. We hide the information of $X$ using one-time-pad coding and we then use the solution of $g_0(P_{XY})=h_0(P_{XY})$ to construct $U$.} 
	\label{kesh11}
\end{figure}
\begin{figure}[]
	\centering
	\includegraphics[scale = .12]{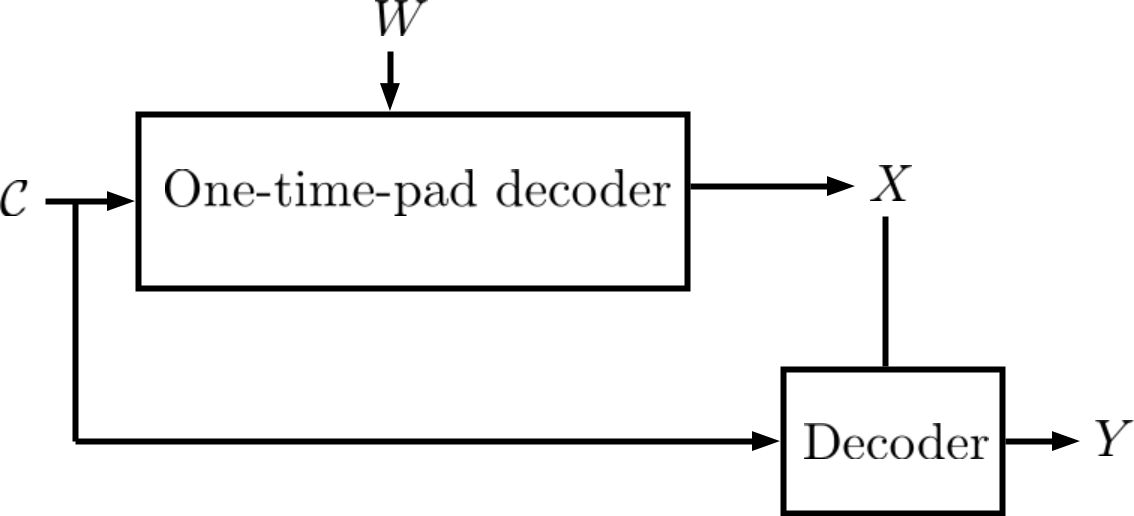}
	\caption{At the receiver side, we first decode $X$ using the shared key $W$, then by using the fact that $U$ satisfies $H(Y|X,U)=0$, we can decode $Y$ based on $X$ and $U$.}
	\label{kesh111}
\end{figure}
\begin{proof}
	All upper bounds can be obtained using two-part construction coding \cite{kostala}. Let $W$ be the shared secret key with key size $M=|\mathcal{X}|$, which is uniformly distributed over $\{1,,\ldots,M\}=\{1,\ldots,|\mathcal{X}|\}$ and independent of $(X,Y)$. As shown in Figure \ref{kesh11}, first, the private data $X$ is encoded using the shared secret key \cite[Lemma~1]{kostala2}. Thus, we have
	\begin{align*}
	\tilde{X}=X+W\ \text{mod}\ |\mathcal{X}|.
	\end{align*}
	Next, we show that $\tilde{X}$ has uniform distribution over $\{1,\ldots,|\mathcal{X}|\}$ and $I(X;\tilde{X})=0$. We have
	\begin{align}\label{s}
	H(\tilde{X}|X)\!=\!H(X\!+\!W|X)\!=\!H(W|X)\!=\!H(W)\!=\! \log(|\mathcal{X}|).
	\end{align}
	Furthermore, $H(\tilde{X}|X)\leq H(\tilde{X})$, and combining it with \eqref{s}, we obtain $H(\tilde{X}|X)= H(\tilde{X})=\log(|\mathcal{X}|)$. For encoding $\tilde{X}$ we use $\ceil{\log(|\mathcal{X}|)}$ bits. Let $\mathcal{C}_1$ denote the encoded $\tilde{X}$. Next, by using $P_{XY}\in\hat{\mathcal{P}}_{XY}$, let $U^*\in\mathcal{U}^1$ be an optimizer of $h_0(P_{XY})=g_0(P_{XY})$ that achieves $\mathcal{K}(P_{XY})$, i.e., has minimum entropy. We encode $U^*$ using any lossless code which uses at most $H(U^*)+1=\mathcal{K}(P_{XY})+1$ bits and let $\mathcal{C}_2$ be the encoded message $U^*$. We send $\mathcal{C}=(\mathcal{C}_1,\mathcal{C}_2)$ over the channel. Thus, \eqref{log} can be obtained. Note that, as shown in Figure \ref{kesh111}, at decoder side, we first decode $X$ using $W$, by adding $|\mathcal{X}|-W$ to $\tilde{X}$, then, $Y$ is decoded using the fact that $U^*$ satisfies $H(Y|U^*,X)=0$. Next, we show that if $W$ is independent of $(X,U^*)$ we get $I(\mathcal{C};X)=0$.
	 We have
	\begin{align*}
	I(\mathcal{C};X)&=I(\mathcal{C}_1,\mathcal{C}_2;X)=I(U,\tilde{X};X)\\&=I(U^*;X)+I(\tilde{X};X|U)=I(\tilde{X};X|U^*)\\&\stackrel{(a)}{=}H(\tilde{X}|U^*)-H(\tilde{X}|X,U^*)\\&=H(\tilde{X}|U^*)-H(X+W|X,U^*)\\&\stackrel{(b)}{=}H(\tilde{X}|U^*)-H(W)\\&\stackrel{(c)}{=}H(\tilde{X})-H(W)=\log(|\mathcal{X}|)\!-\!\log(|\mathcal{X}|)\!=\!0
	\end{align*}
	where (a) follows from $I(U^*;X)=0$; (b) follows since $W$ is independent of $(X,U^*)$; and (c) from the independence of $U^*$ and $\tilde{X}$. The latter follows since we have
	\begin{align*}
	0\leq I(\tilde{X};X|U^*) &= H(\tilde{X}|U^*)-H(W)\\&\stackrel{(i)}{=}H(\tilde{X}|U^*)-H(\tilde{X})\leq 0.
	\end{align*}
	Thus, $\tilde{X}$ and $U^*$ are independent. Step (i) above follows by the fact that $W$ and $\tilde{X}$ are uniformly distributed over $\{1,\ldots,|\mathcal{X}|\}$, i.e., $H(W)=H(\tilde{X})$. As a summary, if we choose $W$ independent of $(X,U^*)$ the leakage to the adversary is zero. Thus, \eqref{log} is achievable. Upper bounds \eqref{ass} and \eqref{mass} can be obtained using the same coding and chain of inequalities stated in Theorem \ref{11}. Upper bounds in \eqref{kos1} and \eqref{govad} have been derived in \cite[Theorem 8]{kostala}. Upper bound in \eqref{kos2} can be obtained using the same coding. Note that the inequality between the upper bounds follows by the fact that the RV $U$ that achieves the upper bound in \eqref{kos1} does not necessarily minimize the entropy $H(U)$, however, it satisfies the constraints $I(X;U)=0$ and $H(Y|X,U)=0$. Finally, to achieve \eqref{kos3}, let the shared key $W$ be independent of $(X,Y)$ and has uniform distribution. We construct $\tilde{Y}$ using one-time pad coding. We have
	\begin{align*}
	\tilde{Y}=Y+W\ \text{mod}\ |\mathcal{Y}|.
	\end{align*}
	Then, $\tilde{Y}$ is encoded using any lossless code which uses at most $\ceil{\log(|\mathcal{Y}|)}$ bits. Latter follows since $\tilde{Y}$ has uniform distribution. It only remains to show that $I(\tilde{Y};X)=0$. We have
	\begin{align*}
	I(\tilde{Y};X,Y)&=I(Y+W;X,Y)\\&=H(Y+W)-H(Y+W|X,Y)\\&=\log(|\mathcal{Y}|)-H(W|X,Y)\\&=\log(|\mathcal{Y}|)-H(W)\\&=\log(|\mathcal{Y}|)-\log(|\mathcal{Y}|)=0.
	\end{align*}
	Hence, $I(\tilde{Y};X)=0$. Furthermore, at decoder side, we can decode $Y$ using $W$.  
\end{proof}
Next we obtain lower bounds on $\mathbb{L}(P_{XY},M)$.
\begin{theorem}
	For any pair of RVs $(X,Y)$ distributed according to $P_{XY}$ and shared key size $M\geq 1$ we have
	\begin{align}\label{converse1}
	\mathbb{L}(P_{XY},M)\geq \max{x}H(Y|X=x),
	\end{align}
	if $X=f(Y)$,
	\begin{align}\label{converse2}
	\mathbb{L}(P_{XY},|\mathcal{X}|)\geq \log(|\mathcal{X}|),
	\end{align}
	and the code $\mathcal{C}$ does not exist when $M< |\mathcal{X}|$ and $X=f(Y)$.
	Furthermore, considering $P_{XY}\in\hat{\mathcal{P}}_{XY}$, if we assume the received code $\mathcal{C}$ satisfies $I(X;\mathcal{C})=0$, $H(Y|X,\mathcal{C})=0$ and $X-Y-\mathcal{C}$ we have
	\begin{align}\label{converse3}
	\mathbb{L}(P_{XY},M)\geq H(Y|X)+\!\!\!\!\!\min_{a_i:A_{XY}\bm{a}=b_{XY},\bm{a}\geq 0}\sum_{i=1}^{q} P_{y_i}a_i.
	\end{align}
\end{theorem}
\begin{proof}
	The proof of \eqref{converse1}, \eqref{converse2}, are provided in \cite[Theorem 9]{kostala} and \eqref{converse3} follows from Theorem \ref{11}.
\end{proof}
\begin{remark}
	Note the constraint $X-Y-\mathcal{C}$ in the converse bound \eqref{converse3} can be motivated by considering the cases where the encoder has no direct access to the private data $X$. Furthermore, the constraint $H(Y|X,\mathcal{C})=0$ is stronger compared to $H(Y|X,\mathcal{C},W)=0$ that is used to prove \eqref{converse1}. More specifically, the constraint needed to prove \eqref{converse1} is $H(Y|\mathcal{C},W)=0$ which results in $H(Y|X,\mathcal{C},W)=0$.
\end{remark}
\section{Comparison:}
In this section, we study the obtained bounds in Theorem \ref{loo} and compare them with the existing results in \cite{kostala}.
\subsection*{Case 1: $|\mathcal{Y}|\leq |\mathcal{X}|$}
Clearly, for any $P_{XY}$ with $|\mathcal{Y}|\leq |\mathcal{X}|$ the upper bound obtained in \eqref{kos3} improves the bounds in \eqref{kos1} and \eqref{kos2} where \eqref{kos1} has been derived in \cite{kostala}. Latter follows since in this case we have $\ceil{\log(|\mathcal{Y}|)}\leq \ceil{\log(|\mathcal{X}|)}$.
\subsection*{Case 2: $P_{XY}\in\hat{\mathcal{P}}_{XY}$}
In this part, let $X=f(Y)$ which results in $P_{XY}\in\hat{\mathcal{P}}_{XY}$, see Corollary 1. The upper bound attained by \cite{kostala} is $\ceil{\log(|\mathcal{Y}|-|\mathcal{X}|+1)}+\ceil{\log(|\mathcal{X}|)}$. Furthemore, we can assume that all elements in probability vector $P_{X}$ are non-zero, otherwise we can remove them. the In this case, we have
\begin{align}
\text{null}(P_{X|Y})=|\mathcal{Y}|-|\mathcal{X}|.
\end{align}
This follows since each column in $P_{X|Y}$ contains exactly one non-zero element that equals to $1$. Moreover, all rows of $P_{X|Y}$ are linearly independent since each row contains a non-zero element that the other rows do not have it. Thus, $\text{rank}(P_{X|Y})=|\mathcal{X}|$ which results $\text{null}(P_{X|Y})=|\mathcal{Y}|-|\mathcal{X}|$. Thus, the upper bound in \eqref{mass} can be modified to get the same bound as \eqref{govad}. Furthermore, the upper bounds \eqref{ass} and \eqref{log} attain less quantities compared to \eqref{govad}, i.e., they can improve the bounds in \cite{kostala}. 
Next, in a numerical example we show that the upper bounds \eqref{ass} and \eqref{log} improve \eqref{govad}.
\begin{example}
	Let $P_{X|Y}=\begin{bmatrix}
	1 &1 &1 &0 &0 &0\\
	0 &0 &0 &1 &1 &1
	\end{bmatrix}$ and $P_Y=[\frac{1}{8},\frac{2}{8},\frac{3}{8},\frac{1}{8},\frac{1}{16},\frac{1}{16} ]$. Clearly, in this case $X$ is a deterministic function of $Y$. Using the linear program proposed in \cite{borz}, we obtain a solution as $P_{Y|u_1}=[0.75,0,0,0.25,0,0]$, $P_{Y|u_2}=[0,0.75,0,0.25,0,0]$, $P_{Y|u_3}=[0,0,0.75,0,0.25,0]$, $P_{Y|u_4}=[0,0,0.75,0,0,0.25]$ and $P_U=[\frac{1}{6},\frac{1}{3},\frac{1}{4},\frac{1}{4} ]$ which results $H(U)=1.9591$ bits. We have $\mathcal{K}(P_{XY})+1\leq 2.9591$ and $\ceil{\log(|\mathcal{Y}|-|\mathcal{X}|+1)}=\ceil{\log{5}}=3$. Hence, the upper bound \eqref{log} is strictly lower than \eqref{govad}, i.e., $U^*$ that achieves $g_0(P_{XY})$ has less entropy than the RV $U$ constructed in \cite[Lemma 1]{kostala}.  
\end{example} 
\section{conclusion}
We have studied a compression problem
with a privacy constraint, where the information delivered
over the channel is
independent of $X$. We propose new upper bounds using two-part
construction coding that benefits from the solution of $g_0(P_{XY})=h_0(P_{XY})$ to encode the second part of the code.
For the first part of the code we hide the private data using
one-time pad coding. Furthermore, in case of $|\mathcal{Y}|\geq |\mathcal{X}|$, we propose a new achievable scheme. We have shown that the new upper bounds can improve the existing ones.

\bibliographystyle{IEEEtran}
\bibliography{IEEEabrv,IZS}
\end{document}